\documentclass[a4paper,conference]{IEEEtran}
\usepackage{graphicx,ifthen,stfloats}
\usepackage{psfrag,amssymb,amsthm}
\usepackage[cmex10]{amsmath}
\usepackage{cite,color,pst-plot,stfloats,pst-sigsys}
\usepackage{flushend}
\usepackage{pstricks-add}

\newtheorem{lem}{Lemma}
\newtheorem{cor}{Corollary}

\newtheorem{prop}{Proposition}

\theoremstyle{definition}
\newtheorem{definition}{Definition}

\newcounter{excnt}

\def \arxiv {1}

\title{On the Rate of Information Loss in Memoryless Systems}

\author{\IEEEauthorblockN{Bernhard C. Geiger\IEEEauthorrefmark{1}, Gernot Kubin\IEEEauthorrefmark{1}
\IEEEauthorblockA{\IEEEauthorrefmark{1}Signal Processing and Speech Communication Laboratory, Graz University of Technology, Austria}
$\{$geiger,g.kubin$\}$@ieee.org}}

\begin{document}
\newcounter{myTempCnt}

% \input{../../Blocks/abbrevations_processing.tex}
% Signals...
\newcommand{\x}[1]{x[#1]}
\newcommand{\y}[1]{y[#1]}

% PDFs
\newcommand{\pdfy}{f_Y(y)}

% Entropies
\newcommand{\overbar}[1]{\mkern 1.5mu\overline{\mkern-3mu#1\mkern-0.5mu}\mkern 1.5mu}

\newcommand{\ent}[1]{H(#1)}
\newcommand{\diffent}[1]{h(#1)}
\newcommand{\derate}[1]{\overbar{h}\left(\mathbf{#1}\right)}
\newcommand{\mutinf}[1]{I(#1)}
\newcommand{\ginf}[1]{I_G(#1)}
\newcommand{\kld}[2]{D(#1||#2)}
\newcommand{\kldrate}[2]{\bar{D}(\mathbf{#1}||\mathbf{#2})}
\newcommand{\binent}[1]{H_2(#1)}
\newcommand{\binentneg}[1]{H_2^{-1}\left(#1\right)}
\newcommand{\entrate}[1]{\overbar{H}(\mathbf{#1})}
\newcommand{\mutrate}[1]{\mutinf{\mathbf{#1}}}
\newcommand{\redrate}[1]{\bar{R}(\mathbf{#1})}
\newcommand{\pinrate}[1]{\vec{I}(\mathbf{#1})}
\newcommand{\loss}[2][\empty]{\ifthenelse{\equal{#1}{\empty}}{L(#2)}{L_{#1}(#2)}}
\newcommand{\lossrate}[2][\empty]{\ifthenelse{\equal{#1}{\empty}}{\overbar{L}(\mathbf{#2})}{L_{\mathbf{#1}}(\mathbf{#2})}}
\newcommand{\gain}[1]{G(#1)}
\newcommand{\atten}[1]{A(#1)}
\newcommand{\relLoss}[2][\empty]{\ifthenelse{\equal{#1}{\empty}}{l(#2)}{l_{#1}(#2)}}
\newcommand{\relLossrate}[1]{l(\mathbf{#1})}
\newcommand{\relTrans}[1]{t(#1)}
\newcommand{\partEnt}[2]{H^{#1}(#2)}

% Domains and Sets
\newcommand{\dom}[1]{\mathcal{#1}}
\newcommand{\indset}[1]{\mathbb{I}\left({#1}\right)}

% Distributions...
\newcommand{\unif}[2]{\mathcal{U}\left(#1,#2\right)}
\newcommand{\chis}[1]{\chi^2\left(#1\right)}
\newcommand{\chir}[1]{\chi\left(#1\right)}
\newcommand{\normdist}[2]{\mathcal{N}\left(#1,#2\right)}
\newcommand{\Prob}[1]{\mathrm{Pr}(#1)}
\newcommand{\Mar}[1]{\mathrm{Mar}(#1)}
\newcommand{\Qfunc}[1]{Q\left(#1\right)}

% Functions...
\newcommand{\expec}[1]{\mathrm{E}\left\{#1\right\}}
\newcommand{\expecwrt}[2]{\mathrm{E}_{#1}\left\{#2\right\}}
\newcommand{\var}[1]{\mathrm{Var}\left\{#1\right\}}
\renewcommand{\det}{\mathrm{det}}
\newcommand{\cov}[1]{\mathrm{Cov}\left\{#1\right\}}
\newcommand{\sgn}[1]{\mathrm{sgn}\left(#1\right)}
\newcommand{\sinc}[1]{\mathrm{sinc}\left(#1\right)}
\newcommand{\e}[1]{\mathrm{e}^{#1}}
\newcommand{\multint}{\iint{\cdots}\int}
\newcommand{\modd}[3]{((#1))_{#2}^{#3}}
\newcommand{\quant}[1]{Q\left(#1\right)}
\newcommand{\card}[1]{\mathrm{card}(#1)}
\newcommand{\diam}[1]{\mathrm{diam}(#1)}
\newcommand{\rec}[1]{r(#1)}
\newcommand{\recmap}[1]{r_{\mathrm{MAP}}(#1)}

% Vectors and Matrices
\newcommand{\ivec}{\mathbf{i}}
\newcommand{\hvec}{\mathbf{h}}
\newcommand{\gvec}{\mathbf{g}}
\newcommand{\avec}{\mathbf{a}}
\newcommand{\kvec}{\mathbf{k}}
\newcommand{\fvec}{\mathbf{f}}
\newcommand{\vvec}{\mathbf{v}}
\newcommand{\xvec}{\mathbf{x}}
\newcommand{\Xvec}{\mathbf{X}}
\newcommand{\Xhvec}{\hat{\mathbf{X}}}
\newcommand{\xhvec}{\hat{\mathbf{x}}}
\newcommand{\xtvec}{\tilde{\mathbf{x}}}
\newcommand{\Yvec}{\mathbf{Y}}
\newcommand{\yvec}{\mathbf{y}}
\newcommand{\Zvec}{\mathbf{Z}}
\newcommand{\Svec}{\mathbf{S}}
\newcommand{\Nvec}{\mathbf{N}}
\newcommand{\Pvec}{\mathbf{P}}
\newcommand{\muvec}{\boldsymbol{\mu}}
\newcommand{\wvec}{\mathbf{w}}
\newcommand{\Wvec}{\mathbf{W}}
\newcommand{\Hmat}{\mathbf{H}}
\newcommand{\Amat}{\mathbf{A}}
\newcommand{\Fmat}{\mathbf{F}}

\newcommand{\zerovec}{\mathbf{0}}
\newcommand{\eye}{\mathbf{I}}
\newcommand{\evec}{\mathbf{i}}

\newcommand{\zeroone}{\left[\begin{array}{c}\zerovec^T\\ \eye\end{array} \right]}
\newcommand{\zerooneT}{\left[\begin{array}{cc}\zerovec & \eye\end{array} \right]}
\newcommand{\zerooneM}{\left[\begin{array}{cc}\zerovec &\zerovec^T\\\zerovec& \eye\end{array} \right]}

\newcommand{\Cxx}{\mathbf{C}_{XX}}
\newcommand{\Cx}{\mathbf{C}_{\Xvec}}
\newcommand{\Chx}{\hat{\mathbf{C}}_{\Xvec}}
\newcommand{\Cy}{\mathbf{C}_{\Yvec}}
\newcommand{\Cz}{\mathbf{C}_{\Zvec}}
\newcommand{\Cn}{\mathbf{C}_{\mathbf{N}}}
\newcommand{\Cnt}{\underline{\mathbf{C}}_{\tilde{\mathbf{N}}}}
\newcommand{\Cntm}{\underline{\mathbf{C}}_{\tilde{\mathbf{N}}}}
\newcommand{\Cxh}{\mathbf{C}_{\hat{X}\hat{X}}}
\newcommand{\rxx}{\mathbf{r}_{XX}}
\newcommand{\Cxy}{\mathbf{C}_{XY}}
\newcommand{\Cyy}{\mathbf{C}_{YY}}
\newcommand{\Cnn}{\mathbf{C}_{NN}}
\newcommand{\Cyx}{\mathbf{C}_{YX}}
\newcommand{\Cygx}{\mathbf{C}_{Y|X}}
\newcommand{\Wmat}{\underline{\mathbf{W}}}

\newcommand{\Jac}[2]{\mathcal{J}_{#1}(#2)}

% Other stuff
\newcommand{\NN}{{N{\times}N}}
\newcommand{\perr}{P_e}
\newcommand{\perh}{\hat{\perr}}
\newcommand{\pert}{\tilde{\perr}}

% Index
% \newcommand{\vecind}[1]{\mathbf{#1}}
\newcommand{\vecind}[1]{#1_0^n}
\newcommand{\roots}[2]{{#1}_{#2}^{(i_{#2})}}
\newcommand{\rootx}[1]{x_{#1}^{(i)}}
\newcommand{\rootn}[2]{x_{#1}^{#2,(i)}}

% Abbrevations
\newcommand{\markkern}[1]{f_M(#1)}
\newcommand{\pole}{a_1}
\newcommand{\preim}[1]{g^{-1}[#1]}
\newcommand{\preimV}[1]{\mathbf{g}^{-1}[#1]}
\newcommand{\Xmax}{\bar{X}}
\newcommand{\Xmin}{\underbar{X}}
\newcommand{\xmax}{x_{\max}}
\newcommand{\xmin}{x_{\min}}
\newcommand{\limn}{\lim_{n\to\infty}}
\newcommand{\limk}{\lim_{k\to\infty}}
\newcommand{\limX}{\lim_{\hat{\Xvec}\to\Xvec}}
\newcommand{\limx}{\lim_{\hat{X}\to X}}
\newcommand{\limXo}{\lim_{\hat{X}_1\to X_1}}
\newcommand{\sumin}{\sum_{i=1}^n}
\newcommand{\finv}{f_\mathrm{inv}}%f_{X_n}^{-1}
\newcommand{\ejtheta}{\e{\jmath\theta}}
\newcommand{\khat}{\bar{k}}
\newcommand{\modeq}[1]{g(#1)}
\newcommand{\partit}[1]{\mathcal{P}_{#1}}
\newcommand{\psd}[1]{S_{#1}(\e{\jmath \theta})}
\newcommand{\borel}[1]{\mathfrak{B}(#1)}
\newcommand{\infodim}[1]{d(#1)}

% signal blocks
\newcommand{\delay}[2]{\psblock(#1){#2}{\footnotesize$z^{-1}$}}
\newcommand{\Quant}[2]{\psblock(#1){#2}{\footnotesize$\quant{\cdot}$}}
\newcommand{\moddev}[2]{\psblock(#1){#2}{\footnotesize$\modeq{\cdot}$}}
\maketitle

\begin{abstract}
In this work we present results about the rate of (relative) information loss induced by passing a real-valued, stationary stochastic process through a memoryless system. We show that for a special class of systems the information loss rate is closely related to the difference of differential entropy rates of the input and output processes. It is further shown that the rate of (relative) information loss is bounded from above by the (relative) information loss the system induces on a random variable distributed according to the process's marginal distribution.

As a side result, in this work we present conditions such that for a Markovian input process also the output process possesses the Markov property.
\end{abstract}


\section{Introduction}\label{sec:Intro}
Signal processing, as defined by many textbooks, is related to the ``representation, transformation, and manipulation of signals and the \emph{information} the signals contain''~\cite[emphasis added]{Oppenheim_Discrete3}. Yet, most of these textbooks leave the notion of \emph{information} completely aside and focus, instead, on purely energetic aspects or second-order statistics: transfer functions for linear filters, their effect on the auto-correlation function of its output signal, and similar results for nonlinear, memoryless systems (e.g.,~\cite{Price_NLGaussianInput}) are popular characterizations. However, except for the purely Gaussian case, energy (or second-order statistics) and information show an inherently different behavior. It is therefore desirable to extend current system theory by information-theoretic aspects.

While the data processing inequality (e.g.,~\cite[p.~35]{Cover_Information2}) captures the fact that deterministic functions of random variables (RVs) destroy information, relatively little has been done to quantify this information loss. Pinsker showed that the entropy rate of a function of a stationary stochastic process on a finite alphabet is bounded from above by the entropy rate of the original process~\cite[Ch.~6.3]{Pinsker_InfoEngl}. Similarily, Watanabe and Abraham analyzed the rate of information loss for functions of stationary stochastic processes, introducing also a relative version of information loss in~\cite{Watanabe_InfoLoss}. Results on the information loss rate in dynamical systems, together with an upper bound, were presented in~\cite{Geiger_NLDyn1starXiv}. 

While these works focus on finite or countable alphabets,~\cite{Geiger_LossLong} analyzes the absolute and, in case the latter is infinite, relative information loss induced by passing a real-valued RV through a memoryless system. In this work we extend~\cite{Geiger_LossLong} to real-valued, stationary stochastic processes. In particular we show that the information loss for RVs distributed according to the marginal distribution of the process is an upper bound on the information loss rate (Section~\ref{sec:LossRate}). A similar result is shown also for the relative information loss rate, although there the bound is tight in many more cases (Section~\ref{sec:RelativeLossRate}): While redundancy helps to reduce the rate of information loss, it often fails to reduce the rate of \emph{relative} informtion loss. The connection between the rate of information loss and the differential entropy rates of the input and output processes shown in Section~\ref{sec:LossRate} is remarkably similar to the corresponding result for information loss presented in~\cite{Geiger_LossLong}.

In search for processes which are simple to analyze, we found a set of sufficient conditions such that for a Markovian input process also the output process has the Markov property (Section~\ref{sec:lumpability}). This extends the notion of \emph{lumpability} (cf.~\cite{Kemeny_FMC}) from discrete-time and continuous-time, homogeneous Markov chains to discrete-time, homogeneous, real-valued Markov processes. These conditions, together with our other theoretical findings, are illustrated with the help of examples in Section~\ref{sec:Examples}.

\ifthenelse{\arxiv=1}{}{Due to the lack of space some of the proofs are omitted in this conference paper, but can instead be found, together with additional examples, in its extended version~\cite{Geiger_LossRate_arXiv}.}

\section{Preliminaries \& Notation}\label{sec:prelim}
Throughout this work we consider discrete-time, stationary stochastic processes $\Xvec$ with alphabet $\dom{X}\subseteq\mathbb{R}$. Let $X_n$ be the $n$-th sample of the process, and let $X_i^j=\{X_i,X_{i+1},\dots,X_j\}$. By stationarity, the distribution of $X_n$, $P_{X_n}$, equals the marginal distribution $P_X$. We assume that $P_X$ is absolutely continuous w.r.t. the Lebesgue measure, and that it thus possesses a probability density function (PDF) $f_X$. Similarily, we assume that for all $n$, the joint PDF $f_{X_1^n}$ and the conditional PDF $f_{X_n|X_1^{n-1}}$ exist.% (w.r.t. the Lebesgue measure on the appropriate Euclidean space). 

Let $\ent{\cdot}$, $\diffent{\cdot}$, $\entrate{\cdot}$, and $\derate{\cdot}$ denote the entropy, the differential entropy, the entropy rate, and the differential entropy rate of the RVs and stochastic processes in the argument (see~\cite{Cover_Information2} or~\cite{Papoulis_Probability} for definitions). We assume that the joint differential entropy of an arbitrary collection of RVs exists and is finite, and that also the entropy rate~\cite[Thm.~14.7]{Papoulis_Probability}
\begin{equation}
 \derate{X} := \limn \diffent{X_n|X_1^{n-1}} = \limn\frac{1}{n}\diffent{X_1^n}\label{eq:derate}
\end{equation}
exists and is finite. The logarithm for the entropies is taken to the base 2.

\section{Information Loss Rate Piecewise Bijective Functions}\label{sec:LossRate}
In this section we devote our attention to a specific class of functions for which the preimage of every point of its range is an at most countable set:

\begin{definition}[Piecewise Bijective Function]\label{def:PBF}
A piecewise bijective function $g{:}\  \dom{X}\to\dom{Y}$, $\dom{X},\dom{Y}\subseteq\mathbb{R}^N$, is a surjective, measurable function defined piecewise on an at most countable partition $\{\dom{X}_i\}$ of its domain:
\begin{equation}
 g(x) = \begin{cases}
  g_1(x), & \text{if } x \in\dom{X}_1\\
  g_2(x), & \text{if } x \in\dom{X}_2\\
  \vdots
 \end{cases}\label{eq:defg}
\end{equation}
where each $g_i{:}\ \dom{X}_i\to\dom{Y}_i$ is bijective. Furthermore, the derivative $g'$ exists on the closures of $\dom{X}_i$, and its magnitude is non-zero $P_X$-a.s.
\end{definition}
Feeding the stationary stochastic process $\Xvec$ through a memoryless system described by such a function $g$ gives rise to another stationary stochastic process $\Yvec$ defined by $Y_n:=g(X_n)$, which, intuitively, conveys less information. In order to analyze the amount of information lost \emph{per sample} we introduce
\begin{definition}[Information Loss Rate]\label{def:lossrate}
 The information loss rate is
\begin{equation}
 \lossrate{X\to Y} := \limn\frac{1}{n} \loss{X_1^n\to Y_1^n} = \limn\frac{1}{n} \ent{X_1^n|Y_1^n}
\end{equation}
i.e., the average of the block information loss.
\end{definition}

We showed in~\cite{Geiger_LossLong} that the information loss in systems described by functions satisfying Definition~\ref{def:PBF} can be computed as
\begin{equation}
 \loss{X\to Y} = \ent{X|Y}=\diffent{X}-\diffent{Y} + \expec{\log|g'(X)|}\label{eq:lossdiffent}
\end{equation}
where $Y=g(X)$ and where the expectation is taken w.r.t. $X$. We now present a corresponding result for stationary stochastic processes:
\begin{prop}[Information Loss Rate for PBFs]\label{prop:lossratePBFs}
The information loss rate induced by feeding a stationary stochastic process $\Xvec$ through a PBF $g$ is
\begin{equation}
 \lossrate{X \to Y} = \derate{X}-\derate{Y}  + \expec{\log|g'(X)|}.\label{eq:lossdiffentRate}
\end{equation}
\end{prop}

\begin{IEEEproof}
 For the proof we note that the $n$ RVs $X_1^n$ can be interpreted as a single, $n$-dimensional RV; similarily, we can define an extended function $g^n{:}\ \dom{X}^n\to\dom{Y}^n$, applying $g$ coordinate-wise. The Jacobian matrix of $g^n$ is a diagonal matrix constituted of the elements $g'(x_i)$. With the extension of~\eqref{eq:lossdiffentRate} to multivariate functions we thus obtain~\cite{Geiger_LossLong}
\begin{multline}
 \loss{X_1^n\to Y_1^n} = \diffent{X_1^n}-\diffent{Y_1^n}+\expec{\log\left|\prod_{i=1}^n g'(X_i)\right|}\\
=\diffent{X_1^n}-\diffent{Y_1^n}+n\expec{\log\left|g'(X)\right|}
\end{multline}
where the first line is because the determinant of a diagonal matrix is the product of its diagonal elements, and where we employed stationarity of $\Xvec$ to obtain the second line. Dividing by $n$ and taking the limit completes the proof.
\end{IEEEproof}

In~\cite{Geiger_LossLong} we showed that the information loss of a cascade of systems equals the sum of the information losses induced in the systems constituting the cascade. Indeed, this result can be carried over to the information loss rate as well:
\begin{prop}[Cascade of Systems]\label{prop:rateCascade}
 Let $\Xvec$ be fed through a PBF $g$ to obtain $\Yvec$, and let $\Yvec$ be fed through a PBF $h$ to obtain $\Zvec$. The information loss rate of the cascade is given as the sum of the individual information loss rates:
\begin{equation}
 \lossrate{X\to Z}=\lossrate{X\to Y}+\lossrate{Y\to Z}.
\end{equation}
\end{prop}

\begin{IEEEproof}
 The proof follows from the fact that the cascade is described by the function $h\circ g$, and that
\begin{multline}
 \expec{\log|(h\circ g)'(X)|} = \expec{\log|g'(X)h'(g(X))|}\\
=\expec{\log|g'(X)|}+\expec{\log|h'(Y)|}.
\end{multline}
\end{IEEEproof}

It is often not possible to obtain closed-form expressions for the information loss rate induced by a system. Moreover, estimating the information loss rate by simulations soon suffers the curse of dimensionality, as, in principle, infinitely long random sequences have to be drawn and averaged. Much simpler is an estimation of the information loss, since a single realized, sufficiently long sequence allows for an estimation of the latter. As the next proposition shows, this relatively simple estimation delivers an upper bound on the information loss rate:
\begin{prop}[Loss $>$ Loss Rate]\label{prop:lossgelossrate}
 Let $\Xvec$ be a stationary stochastic process and $X$ an RV distributed according to the process's marginal distribution. The information loss induced by feeding $X$ through a PBF $g$ is an upper bound on the information loss rate induced by passing $\Xvec$ through $g$, i.e.,
\begin{equation}
 \lossrate{X\to Y} \le \loss{X\to Y}.
\end{equation}
\end{prop}

\ifthenelse{\arxiv=1}{
\begin{IEEEproof}
The inequality holds trivially if $\loss{X\to Y}=\infty$. The rest of the proof follows from the chain rule and the fact that conditioning reduces entropy:
\begin{IEEEeqnarray}{RCL}
 \lossrate{X\to Y} &=& \limn \frac{1}{n} \ent{X_1^n|Y_1^n}\\
&=& \limn \frac{1}{n} \sum_{i=1}^n \ent{X_i|X_{1}^{i-1},Y_1^n}\\
&\le& \limn \frac{1}{n} \sum_{i=1}^n \ent{X_i|Y_i}\\
&=& \loss{X\to Y}.
\end{IEEEeqnarray}
\end{IEEEproof}}
{The proof employs the chain rule and the fact that conditioning reduces entropy and is given in~\cite{Geiger_LossRate_arXiv}.}
Clearly, this bound is tight whenever the input process $\Xvec$ is an iid process. Moreover, it is trivially tight whenever the function is bijective, i.e., when $\loss{X\to Y}=0$. In \ifthenelse{\arxiv=1}{Section~\ref{ssec:exampleTight}}{\cite{Geiger_LossRate_arXiv}} we present an example which renders this bound tight in the general case.

Intuitively, this bound suggests that redundancy of a process, i.e., the statistical dependence of its samples, reduces the amount of information lost \emph{per sample} when fed through a deterministic system. The same connection between information loss and information loss rate has already been observed in~\cite{Watanabe_InfoLoss} for stationary stochastic processes with finite alphabets.

The next bound again extends a result from~\cite{Geiger_LossLong}, bounding the information loss rate by the entropy rate of a stationary stochastic process on an at most countable alphabet. As such, it presents a different way to estimate the information loss rate efficiently using numerical simulations.
\begin{prop}[Upper Bound]\label{prop:WBound}
 Let $\Wvec$ be a stationary stochastic process defined by $W_n:=i$ if $X_n\in\dom{X}_i$. Then, 
\begin{equation}
 \lossrate{X \to Y} = \entrate{W|Y} \le \entrate{W}.
\end{equation}
\end{prop}

\begin{proof}
 We again treat $X_1^n$ as an $n$-dimensional RV; $g^n$ induces a partition of its domain $\dom{X}^n$, which is equivalent to the $n$-fold product of the partition $\{\dom{X}_i\}$. Letting $\tilde{W}$ be the RV obtained by quantizing $X_1^n$ according to this partition, it is easy to see that $W_1^n$ is equivalent to $\tilde{W}$. Thus, with~\cite{Geiger_LossLong},
\begin{equation}
 \ent{X_1^n|Y_1^n} = \ent{\tilde{W}|Y_1^n} = \ent{W_1^n|Y_1^n}
\end{equation}
for all $n$. This, together with the fact that conditioning reduces entropy, completes the proof.
\end{proof}

For the case that the input process is a Markov process, i.e., if $f_{X_n|X_1^{n-1}}=f_{X_n|X_{n-1}}$ for all $n$, an additional, sharper, upper bound can be presented:
\begin{prop}[Upper Bound for Markovian $\Xvec$]\label{prop:UBMarkov}
 Let $\Xvec$ be a Markov process, and let $\Wvec$ be as in Proposition~\ref{prop:WBound}. Then, for finite $\loss{X \to Y}$,
\begin{equation}
 \lossrate{X \to Y} \le \ent{W_2|X_1}.
\end{equation}
\end{prop}

\begin{IEEEproof}
 We again apply the chain rule, Markovity of $\Xvec$, and the fact that conditioning reduces entropy to arrive at
\begin{equation}
  \ent{X_1^n|Y_1^n} \le\sum_{i=1}^n \ent{X_i|X_{i-1},Y_i}.
\end{equation}
By stationarity we obtain
\begin{IEEEeqnarray}{RCL}
 \lossrate{X \to Y} &\le& \ent{X_2|X_{1},Y_2}\label{eq:Markovity}\\
&\stackrel{(a)}{=}& \ent{W_2|X_{1},Y_2}\\
&\le& \ent{W_2|X_{1}}\label{eq:prop52ndeq}
\end{IEEEeqnarray}
where $(a)$ holds since, for all $x\in\dom{X}$, $\ent{X_2|Y_2,X_1=x}=\ent{W_2|Y_2,X_1=x}$~\cite{Geiger_LossLong}. The last inequality is due conditioning~\cite[Thm.~2.6.5]{Cover_Information2} and completes the proof.
\end{IEEEproof}
That the bound is sharper than the one of Proposition~\ref{prop:WBound} follows from observing that
\begin{multline}
 \ent{W_n|X_{n-1}}=\limn\ent{W_n|X_1^{n-1}}\\\le\limn\ent{W_n|W_1^{n-1}}=\entrate{W}.
\end{multline}

The interpretation of this result is that a function destroys little information if the process is such that, given the current sample $X_{n-1}$, the next sample $X_n$ falls within some element of the partition with a high probability. \ifthenelse{\arxiv=1}{The question whether, and under which conditions, this bound is tight is related to the phenomenon of lumpability and will be answered in the following section.}{The question whether, and under which conditions,~\eqref{eq:Markovity} becomes an equality is related to the phenomenon of lumpability and will be answered in the following section. Conditions for tightness of this bound are given and illustrated with an example in~\cite{Geiger_LossRate_arXiv}.}

\section{Lumpability for Continuous-Valued Markov Processes}\label{sec:lumpability}
It is well-known that the function of a Markov process need not possess the Markov property itself. However, as it is known for Markov chains, there exist conditions on the function and/or the chain such that the output is Markov. In~\cite{Kemeny_FMC} this has been termed \emph{lumpability} and subsequently investigated by numerous researchers. While most results are given for finite Markov chains (e.g.,~\cite{Burke_MarkovFunction,Rogers_MarkovFunct}) relatively little is known in the general case of an uncountable alphabet (see~\cite{Rosenblat_Markov} for an exception). Our small contribution to this field of research lies in presenting sufficient conditions for lumpability of continuous-valued Markov processes. \ifthenelse{\arxiv=1}{}{The proofs can be found in~\cite{Geiger_LossRate_arXiv}.}

Let $f_{X_n|X_1^{n-1}}=f_{X_n|X_{n-1}}$ for all $n$, i.e., let $\Xvec$ be a Markov process. We maintain
\begin{prop}\label{prop:equalityThenLump}
 If 
\begin{multline}
 \forall y_1^2\in\dom{Y}^2: \forall x\in\preim{y_1}:\\
 f_{Y_2,X_1}(y_2,x)>0 \Rightarrow f_{Y_2|X_1}(y_2|x) = f_{Y_2|Y_1}(y_2|y_1)
\end{multline}
then $\Xvec$ is lumpable w.r.t. $g$, i.e., $\Yvec$ is Markov.
\end{prop}
\ifthenelse{\arxiv=1}
{\begin{IEEEproof}
 See Appendix.
\end{IEEEproof}
As a corollary, we next make the conditions on the function $g$, the marginal distribution $f_X$, and the conditional distribution $f_{X_2|X_1}$ explicit. By adding a further condition, we gain tightness of Proposition~\ref{prop:UBMarkov} in addition to Markovity:
}
{A direct consequence of this proposition is $\derate{Y}=\diffent{Y_2|Y_1}=\diffent{Y_2|X_1}$. As a corollary, we make the conditions on the function $g$, the marginal distribution $f_X$, and the conditional distribution $f_{X_2|X_1}$ explicit.}
\begin{cor}\label{cor:sufficientConditions}
 If for all $y_1^2\in\dom{Y}^2$ and all $x,x'\in\preim{y_1}$ such that $f_{X}(x)>0$ and $f_{X}(x')>0$ the following holds
% \begin{subequations}
\begin{align}
%  \frac{f_X(x)}{|g'(x)|} &= \frac{f_X(x')}{|g'(x')|}\label{eq:margEq}\\
 \sum_{x_2\in\preim{y_2}} \frac{f_{X_2|X_1}(x_2|x)}{|g'(x_2)|} &= \sum_{x_2\in\preim{y_2}} \frac{f_{X_2|X_1}(x_2|x')}{|g'(x_2)|} \label{eq:condEq}
\end{align}
% \end{subequations}
then the condition of Proposition~\ref{prop:equalityThenLump} is fulfilled and $\Yvec$ is Markov. 

\ifthenelse{\arxiv=1}{If, additionally, for all $y\in\dom{Y}$, all $x$ within the support of $f_X$, and all $w,w'$ such that $\Prob{W_2=w|X_1=x}>0$ and $\Prob{W_2=w'|X_1=x}>0$
\begin{subequations}\label{eq:tightEq}
 \begin{equation}
 \frac{f_{X_2|X_1}(g_w^{-1}(y_2)|x)}{|g'(g_w^{-1}(y_2))|} = \frac{f_{X_2|X_1}(g_{w'}^{-1}(y_2)|x)}{|g'(g_{w'}^{-1}(y_2))|}\label{eq:tightEq:a}
\end{equation}
and
\begin{equation}
 \Prob{W_2=w'|X_1=x}=\Prob{W_2=w|X_1=x}\label{eq:tightEq:b}
\end{equation}
\end{subequations}
then the bound of Proposition~\ref{prop:UBMarkov} holds with equality.}{}
\end{cor}

\ifthenelse{\arxiv=1}{
\begin{IEEEproof}
 See Appendix.
\end{IEEEproof}
In Section~\ref{sec:Examples} we show some examples for which the output process $\Yvec$ is Markov and for which the conditions in~\eqref{eq:tightEq} are fulfilled.}{In Section~\ref{ssec:ARProcess} we present an example for which the output process $\Yvec$ is Markov.}

\section{Relative Information Loss Rate for Functions which Reduce Dimensionality}\label{sec:RelativeLossRate}
Not all systems can be described by functions satisfying Definition~\ref{def:PBF}. In particular, a simple quantizer already violates this definition and suffers from infinite information loss. To analyze the information processing characteristics of a broader class of systems, in~\cite{Geiger_LossLong} the notion of relative information loss was introduced, capturing the \emph{percentage} of information available at the input lost in the system. To extend this notion to stochastic processes, we introduce
\begin{definition}[Relative Information Loss Rate]\label{def:RelLossRate}
 The relative information loss rate is
\begin{equation}
 \relLossrate{X\to Y} := \limn \relLoss{X_1^n\to Y_1^n} = \limn \lim_{\hat{\Xvec}\to\Xvec}\frac{\ent{\hat{X}_1^n|Y_1^n}}{\ent{\hat{X}_1^n}}
\end{equation}
whenever the limit exists.
\end{definition}
The limit $\hat{\Xvec}\to\Xvec$ is equivalent to $\lim_{k\to\infty} \lfloor 2^k \Xvec\rfloor/2^k$, where flooring and scalar multiplication are applied element-wise (cf.~\cite{Geiger_LossLong}).

\ifthenelse{\arxiv=1}{Based on
\begin{equation}
 \relLoss{X_1^n\to Y_1^n} \le \frac{1}{n}\sum_{i=1}^n \relLoss{X_i\to Y_i} \stackrel{(a)}{=} \relLoss{X\to Y}
\end{equation}
from~\cite{Geiger_LossLong} and from stationarity of $\Xvec$, which yields $(a)$, one can show\footnote{Note that also Watanabe and Abraham defined the \emph{fractional information loss} for stochastic processes on finite alphabets~\cite{Watanabe_InfoLoss}; for these types of processes, however, the relative information loss can be smaller \emph{or} larger than the information loss.} that $\relLossrate{X\to Y} \le \relLoss{X\to Y}$, complementing Proposition~\ref{prop:lossgelossrate}. 
}{Based on the upper bound for relative information loss in~\cite{Geiger_LossLong}, we can show that also $\relLossrate{X\to Y} \le \relLoss{X\to Y}$, complementing Proposition~\ref{prop:lossgelossrate}.} However, in many cases this inequality is an equality, as we show in
\begin{prop}[Redundancy won't help]\label{prop:RelLossgeRelLossRate}
  Let $\Xvec$ be a stationary stochastic process and $X$ an RV distributed according to the process's marginal distribution. Let further $g$ be defined on a finite partition $\{\dom{X}_i\}$ of $\dom{X}$ into non-empty sets as in~\eqref{eq:defg}, where $g_i\in\mathcal{C}^\infty$ is either injective or constant (i.e., $g_i(x)=c_i$ for all $x\in\dom{X}_i$). Then,
\begin{equation}
 \relLossrate{X\to Y} = \relLoss{X\to Y} = P_X(\dom{X}_c)
\end{equation}
where $\dom{X}_c$ is the union of all elements $\dom{X}_i$ of the partition on which $g$ is constant.
\end{prop}

\begin{IEEEproof}
\ifthenelse{\arxiv=1}{See Appendix.}{See~\cite{Geiger_LossRate_arXiv}.}
\end{IEEEproof}

Indeed, we conjecture that equality is indeed the ``usual'' case, prevailing in most practical scenarios. Thus, while redundancy can help reduce information loss, it may be useless when it comes to relative information loss. Applications of this result may be the scalar quantization of a stochastic process (leading to a relative information loss rate of 1, i.e., 100\% of the information is lost~\cite{Geiger_LossLong}) and system blocks for multirate signal processing (see the example in Section~\ref{ssec:multirate}).

\section{Examples}\label{sec:Examples}
\subsection{AR-Process and Magnitude Function}\label{ssec:ARProcess}
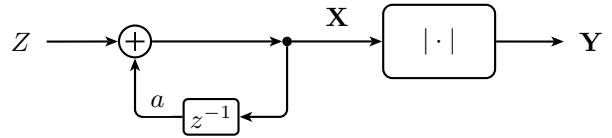
\begin{figure}[t]
 \centering
\begin{pspicture}[showgrid=false](1.5,1)(9,2.5)
	\psset{style=Arrow,style=RoundCorners}
 	\pssignal(1.5,2){z}{$Z$}
	\pscircleop(3,2){oplus}
	\dotnode(5,2){dot}
	\psfblock[framesize=1.5 1](7,2){c}{$|\cdot|$}
	\psfblock[framesize=0.75 0.5](4,1){del}{$z^{-1}$}
	\ncangle[angleA=-90,angleB=0]{dot}{del}
	\ncangle[angleA=180,angleB=-90]{del}{oplus}
	\nbput*{$a$}
	\pssignal(9,2){y}{$\Yvec$}
	\nclist[style=Arrow]{ncline}[naput]{z,oplus,dot,c $\Xvec$ ,y}
\end{pspicture}
 \caption{AR(1)-process with magnitude function. The input $Z$ is a sequence of iid Gaussian RVs with zero mean and variance $\sigma^2$; thus, the process $\Xvec$ is Gaussian with zero mean and variance $\sigma^2/1-a^2$. The process generator filter is a first-order all-pole filter with a single pole at $a$.}
 \label{fig:ARSysModel}
\end{figure}

\newcommand{\normpdf}[1]{\phi(#1)}

In this example we assume that a first-order, zero-mean, Gaussian auto-regressive process $\Xvec$ is fed through a magnitude function (see Fig.~\ref{fig:ARSysModel}). Let the AR process be generated by the following difference equation:
\begin{equation}
 X_n=aX_{n-1}+Z_n
\end{equation}
where $a\in (0,1)$ and where $Z_n$ are samples drawn independently from a Gaussian distribution with zero mean and variance $\sigma^2$. It follows immediately that the process $\Xvec$ is also zero mean and has variance $\sigma_X^2=\frac{\sigma^2}{1-a^2}$~\cite[Ex.~6.11]{Oppenheim_Discrete3}. Let $\Yvec$ be defined by $Y_n=|X_n|$.

For the sake of brevity we define $\normpdf{\mu,\sigma^2;x}$ as the PDF of a Gaussian RV with mean $\mu$ and variance $\sigma^2$, evaluated at $x$. Thus, we get
\begin{equation}
 f_X(x) = \normpdf{0,\sigma_X^2;x}
\end{equation}
and
\begin{equation}
 f_{X_2|X_1}(x_2|x_1) = \normpdf{ax_1,\sigma^2;x_2}.
\end{equation}
It follows that~\eqref{eq:condEq} is satisfied with $|g'(x)|\equiv 1$ and since $\normpdf{ax_1,\sigma^2;x_2}=\normpdf{-ax_1,\sigma^2;-x_2}$,
\begin{multline}
 \sum_{x_2\in\preim{y_2}} \frac{f_{X_2|X_1}(x_2|y_1)}{|g'(x_2)|}\\
=\normpdf{ay_1,\sigma^2;y_2}+\normpdf{ay_1,\sigma^2;{-}y_2}\\
=\normpdf{{-}ay_1,\sigma^2;{-}y_2}+\normpdf{{-}ay_1,\sigma^2;y_2}\\
=\sum_{x_2\in\preim{y_2}} \frac{f_{X_2|X_1}(x_2|{-}y_1)}{|g'(x_2)|}.
\end{multline}
As a consequence, the output process $\Yvec$ is Markov.

We performed a series of simulations, as the information loss rate for this example cannot be expressed in closed form. Rewriting, e.g., the lower bound on the information loss rate as \ifthenelse{\arxiv=1}{}{(see~\cite{Geiger_LossRate_arXiv} for an explanation)}
\begin{align}
& \lossrate{X\to Y}\notag\\ &\ge \diffent{X_2|X_1}-\diffent{Y_2|X_1}+\expec{\log|g'(X)|}\\
&= \diffent{X}-\mutinf{X_1;X_2}-\diffent{Y}+\mutinf{X_1;Y_2}\notag\\&\quad{}+\expec{\log|g'(X)|}\\
&= \loss{X\to Y}-\mutinf{X_1;X_2}+\mutinf{X_1;Y_2}
\end{align}
allowed us to employ the histogram-based mutual information estimation from~\cite{Moddemeijer_Matlab} together with $\loss{X\to Y}=1$, as shown in~\cite{Geiger_LossLong}. The upper bound $H(W_2|X_1)$ from Proposition~\ref{prop:UBMarkov} was computed using numerical integration. In Fig.~\ref{fig:ARInfoLoss} one can see that the first-order upper and lower bounds on the information loss rate \ifthenelse{\arxiv=1}{from Lemma~\ref{lem:boundsDERate} in the proof of Proposition~\ref{prop:UBMarkov}}{(see the proof of Proposition~\ref{prop:UBMarkov} in~\cite{Geiger_LossRate_arXiv})} are indistinguishable, which suggests that the output process is indeed Markov. Moreover, it can be seen that a higher value for the magnitude $a$ of the pole leads to a smaller information loss rate. This can be explained by the fact that the redundancy\footnote{The redundancy is defined as the difference between the entropy of the marginal distribution and the entropy rate of the process. The former increases due to increasing variance $\sigma^2_X$, while the latter remains constant and equal to $\diffent{Z}$ (cf.~\cite{Dumitrescu_EntropyInvariance}).} of the process $\Xvec$ increases with increasing $a$, which helps preventing information loss.

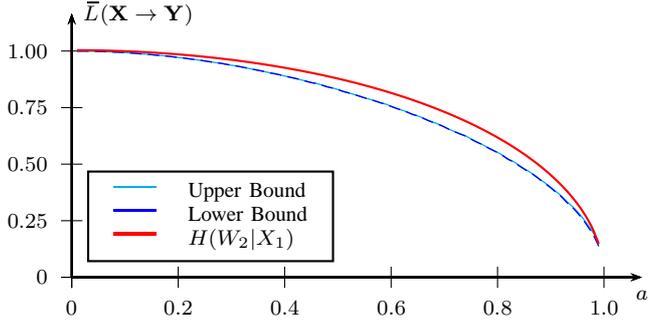
\begin{figure}[t]
 \centering
\begin{pspicture}[showgrid=false](-1,-0.5)(8,3.6)
	\footnotesize
	\psaxes[Dx=0.2,dx=1.4,Dy=0.25,dy=0.75]{->}(0,0)(7.5,3.5)[$a$,-90][$\lossrate{X\to Y}$,0]
	\rput[lb](0.2,0.2){\psframebox%
	{\begin{tabular}{ll}
	  	\cyan\rule[1ex]{2em}{0.5pt} & Upper Bound\\%
		\blue\rule[1ex]{2em}{0.5pt} & Lower Bound\\%
		\red\rule[1ex]{2em}{1pt} & $\ent{W_2|X_1}$%
	 \end{tabular}}}
	\readdata{\UB}{AR1UB.dat}
	\readdata{\LB}{AR1LB.dat}
	\readdata{\HWX}{AR1HWX.dat}
	\psset{xunit=7cm,yunit=3cm}
% 	\begin{psgraph}(0,0)(7.5,3.5){7.5cm}{3.5cm}
	 	\dataplot[plotstyle=curve,linecolor=cyan,linewidth=0.5pt]{\UB}
		\dataplot[plotstyle=curve,linecolor=blue,style=Dash,linewidth=0.5pt]{\LB}
		\dataplot[plotstyle=curve,linecolor=red]{\HWX}
% 	\end{psgraph}
\end{pspicture}
 \caption{Information Loss Rate of an AR(1)-process $\Xvec$ in a magnitude function as a function of the pole $a$ of the process generator difference equation. A larger pole, leading to a higher redundancy of $\Xvec$, reduces the information loss rate.}
 \label{fig:ARInfoLoss}
\end{figure}

Generally, while redundancy reduces the information loss rate compared to an iid process (cf. Proposition~\ref{prop:lossgelossrate}), it is not necessarily true that \emph{more} redundancy leads to a smaller information loss rate than \emph{less} redundancy. Indeed, one can generate examples where a process with a higher redundancy suffers from a higher information loss rate than a process with less redundancy. This suggests that the redundancy of a process has to be matched to the function $g$ in order to efficiently prevent information from being lost; in that sense, this parallels the field of channel coding, where the code needs to be matched to the characteristics of the channel (noise, fading, burst errors) in order to successfully reduce the bit error rate.

\ifthenelse{\arxiv=0}{}{
\subsection{Cyclic Random Walk}\label{ssec:randomwalk}
We next consider a scenario where our process is a cyclic random walk on a subset $[-M,M]$ of the real line. Assume that for a given state $X_1$ the following state is uniformly distributed on a cyclically shifted subset of $[-M,M]$ of length $2a\le 2M$, i.e.,
\begin{equation}
 f_M(x_2|x_1):=f_{X_2|X_1}(x_2|x_1) = \begin{cases}
                         \frac{1}{2a}, & \text{if } d(x_2,x_1)\le a\\0,&\text{else}
                        \end{cases}
\end{equation}
where $d(x,y)=\min_k |x-y-2kM|$. Intuitively, $X_n$ is the sum of $n$ independent RVs uniformly distributed on $[-a,a]$, where sums outside of $[-M,M]$ are mapped back into this interval via the modulo operation. It is easy to verify that the marginal distribution of $\Xvec$ is the uniform distribution\footnote{The discrete-valued equivalent is a Markov chain with a doubly stochastic transition matrix, for which it is known that the stationary distribution is the uniform distribution~\cite[p.~732]{Papoulis_Probability}.}, i.e., $f_X(x)=\frac{1}{2M}$ for all $x\in [-M,M]$ and zero otherwise. The function we feed the process through shall again be the magnitude function, i.e., $Y_n=|X_n|$.

Since $d(x,y)=d(-x,-y)$ and since $|g'(x)|\equiv 1$ for all $x$, it follows that~\eqref{eq:condEq} is fulfilled, and that thus $\Yvec$ is Markov. Moreover, we have $\derate{Y}=\diffent{Y_2|X_1}$, and obtain for the information loss rate with Proposition~\ref{prop:lossratePBFs}
\begin{align}
 &\lossrate{X\to Y}\notag\\
&= \derate{X}-\derate{Y}+\expec{\log|g'(X)|}\\
&= \diffent{X_2|X_1}-\diffent{Y_2|X_1}\\
&= \int_{-M}^M\int_{-M}^M f_X(x_1)  f_{M}(x_2|x_1) \log \frac{f_{Y_2|X_1}(|x_2||x_1)}{f_{M}(x_2|x_1)}dx_2dx_1\\
&=\int_{-M}^M\int_{-M}^M \frac{f_{M}(x_2|x_1)}{2M} \log \left(1+\frac{f_{M}(-x_2|x_1)}{f_{M}(x_2|x_1)}\right)dx_2dx_1.
\end{align}
The logarithm evaluates to zero if $f_{M}(-x_2|x_1)=0$ and to one otherwise (the logarithm is taken to base 2). Therefore, we can write
\begin{align}
 &\lossrate{X\to Y}\notag\\
% &=\frac{2a}{2M}\int_{-M}^M\int_{-M}^M f_{M}(x_2|x_1)f_{M}(-x_2|x_1)dx_2dx_1\\
&=\frac{4a}{2M}\int_0^M\int_{-M}^M f_{M}(x_2|x_1)f_{M}(-x_2|x_1)dx_2dx_1
\end{align}
where we exploited the symmetry of $f_M$. It can be shown that the integral evaluates to $\frac{1}{2}$, so the information loss rate is $\lossrate{X\to Y}=\frac{a}{M}$. 

This result has a nice geometric interpretation: It quantifies the expected overlap of two segments of length $2a$ randomly placed on a circle with circumference $2M$; due to the modulo operation the point $-M$ is equivalent to the point $M$, and the conditional PDFs $f_M(x_2|x_1)$ and $f_M(-x_2|x_1)$ represent the segments (see Fig.~\ref{fig:CycInterpretation}).

Finally, we evaluated the upper bound from Proposition~\ref{prop:UBMarkov}: Letting $\dom{X}_1=[-M,0)$ and $\dom{X}_2=[0,M]$ and abbreviating $p(1|x):=\Prob{W_2=1|X_1=x}$ we obtained
\begin{equation}
 p(1|x)=
\begin{cases}
	\frac{a-M-x}{2a},& -M\le x<-M+a\\
	0,& -M+a\le x<-a\\
	\frac{x+a}{2a},& -a\le x<a\\
	1,& a\le x <M-a\\
	\frac{a+M-x}{2a},& M-a\le x<M
\end{cases}
\end{equation}
if $M>2a$ and
\begin{equation}
 p(1|x)=
\begin{cases}
	\frac{a-M-x}{2a},& -M\le x<-a\\
	\frac{2a-M}{2a},& -a\le x<-M+a\\
	\frac{x+a}{2a},& -M+a\le x<M-a\\
	\frac{M}{2a},& M-a\le x <a\\
	\frac{a+M-x}{2a},& a\le x<M
\end{cases}
\end{equation}
if $M\le2a$. (Naturally, $p(2|x)=1-p(1|x)$.) Computing the entropy $\ent{W_2|X_1=x}$ based on these probabilities and taking the expectation w.r.t. $X_1$ yields
\begin{equation}
 \ent{W_2|X_1}=
\begin{cases}
 \frac{a}{M\ln2},&M>2a\\
 \frac{M-a}{M\ln2}+\log\frac{2a}{M},&M\le 2a
\end{cases}.
\end{equation}
The analytic result for the information loss rate and the bound, numerically validated using the same procedure as in Section~\ref{ssec:ARProcess}, are depicted in Fig.~\ref{fig:CycLoss}.

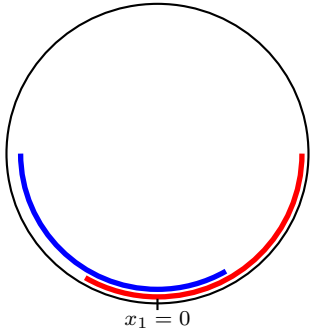
\begin{figure}[t]
 \centering
\begin{pspicture}[showgrid=false](-4,-2)(4,2)
	\footnotesize
	\pscircle(0,0){2}
	\psarc[linecolor=red,linewidth=2pt](0,0){1.9}{-120}{0}
	\psarc[linecolor=blue,linewidth=2pt](0,0){1.8}{-180}{-60}
	\psTick{90}(0,-2)
	\rput(0,-2.2){$x_1=0$}
\end{pspicture}
 \caption{Interpreting the information loss rate of a cyclic random walk in a magnitude function. The depicted scenario corresponds to $a=M/3$.}
 \label{fig:CycInterpretation}
\end{figure}

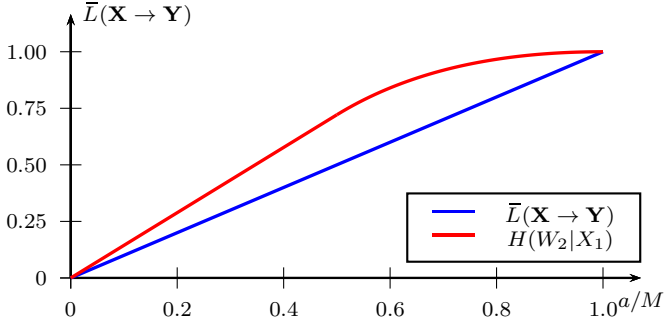
\begin{figure}[t]
 \centering
\begin{pspicture}[showgrid=false](-1,-0.5)(8,3.6)
	\footnotesize
	\psaxes[Dx=0.2,dx=1.4,Dy=0.25,dy=0.75]{->}(0,0)(7.5,3.5)[$a/M$,-90][$\lossrate{X\to Y}$,0]
	\rput[rb](7.5,0.2){\psframebox%
	{\begin{tabular}{ll}
	  	\blue\rule[1ex]{2em}{0.5pt} & $\lossrate{X\to Y}$\\%
		\red\rule[1ex]{2em}{1pt} & $\ent{W_2|X_1}$%
	 \end{tabular}}}
	\psplot[style=Graph,linecolor=blue]{0}{7}{x 7 div 3 mul}
	\psplot[style=Graph,linecolor=red]{0}{3.5}{x 7 div 3 mul 2 ln div}
	\psplot[style=Graph,linecolor=red]{3.5}{7}{3 2 ln div x 7 div 3 mul 2 ln div sub x 7 div 2 mul ln 2 ln div 3 mul add}
\end{pspicture}
 \caption{Information loss rate of a cyclic random walk $\Xvec$ on $[-M,M]$ in a magnitude function as a function of the support $[-a,a]$ of the uniform input PDF.}
 \label{fig:CycLoss}
\end{figure}

\subsection{An Example illustrating the Tightness of the Bounds}\label{ssec:exampleTight}
The following example illustrates the tightness of the presented bounds and also satisfies the condition of lumpability. Assume that $\Xvec$ is a Markov process with conditional distribution
\begin{equation}
 f_M(x_2|x_1) = \frac{1}{2}
\begin{cases}
   \mathbb{I}_{[1,2)}(x_2)+\mathbb{I}_{[3,4)}(x_2), & x_1\in[0,1)\cup[2,3)\\
   \mathbb{I}_{[0,1)}(x_2)+\mathbb{I}_{[2,3)}(x_2), & x_1\in[1,2)\cup[3,4)
\end{cases}
\end{equation}
where $\mathbb{I}_A(x)=1$ iff $x\in A$. As it can be shown easily, the stationary distribution is the uniform distribution on $[0,4)$, thus, $\diffent{X}=\log 4=2$ and $\derate{X}=1$.

We analyze the system mapping the interval $\dom{X}_2=[2,4)$ onto the interval $\dom{X}_1=[0,1)$, i.e., 
\begin{equation}
 g(x)=\begin{cases}
       x,& x\in[0,2)\\ x-2,& x\in[2,4)
      \end{cases}
\end{equation}
which yields the conditional distribution of $Y_2$ given $X_1$
\begin{equation}
 f_{Y_2|X_1}(y_2|x_1) = 
\begin{cases}
   \mathbb{I}_{[1,2)}(y_2), & x_1\in[0,1)\cup[2,3)\\
   \mathbb{I}_{[0,1)}(y_2), & x_1\in[1,2)\cup[3,4)
\end{cases}.
\end{equation}
The derivative of this function is identical to one, the stationary distribution of the output process $\Yvec$ is the uniform distribution on $[0,2)$; thus, $\diffent{Y}=\log 2=1$ and $\loss{X\to Y}=1$.

The output process $\Yvec$ can be shown to be Markov: Assuming $X_1=x\in[0,1)$, it follows that $x'\in[2,3)$; since these conditions are equivalent in the definition of $f_M$,~\eqref{eq:condEq} is fulfilled.

From $f_M$ one can see that $\Prob{W_2=1|X_1=x}=\Prob{W_2=2|X_1=x}=\frac{1}{2}$ regardless of $x$, which satisfies~\eqref{eq:tightEq:b} and renders the upper bound from Proposition~\ref{prop:UBMarkov} as
\begin{equation}
 \ent{W_2|X_1}=1.
\end{equation}
The bound can be shown to be tight, since also~\eqref{eq:tightEq:a} is fulfilled: Given, e.g., $X_1=x\in[0,1)$ and $Y_2=y\in[1,2)$, it follows that $X_2\in\{y,y+2\}$ and $f_M(y|x)=f_M(y+2|x)=\frac{1}{2}$. Thus, we are led to the following conclusion:
\begin{equation}
 1=\loss{X\to Y}\ge \lossrate{X\to Y}=\ent{W_2|X_1}=1
\end{equation}
This is an example not only for tightness of Proposition~\ref{prop:UBMarkov} but also of Proposition~\ref{prop:lossgelossrate}. Interestingly, neither is the function information-preserving, nor is the input process $\Xvec$ iid. Consequently, one can interpret this example as a worst-case, where redundancy is not matched to the system (the ``channel''), failing to alleviate the adverse effects of the system.
}

\subsection{Multirate Systems}\label{ssec:multirate}
Although strictly speaking not time-invariant, also multirate systems can be analyzed with the proposed relative information loss rates. In particular, we will show that for an $M$-fold downsampler, which is described by the input-output relation $Y_n = X_{nM}$, the information loss rate equals
\begin{equation}
 \relLossrate{X\to Y} = \frac{M-1}{M}.
\end{equation}
To this end, note that the stationary output process $\Yvec$ is equivalent to the cyclo-stationary process $\tilde{\Yvec}$, whose samples are defined as
\begin{equation}
 \tilde{Y}_n=
\begin{cases}
	X_n, & \text{if } n/M \in\mathbb{Z}\\
	0,& \text{else }
\end{cases}.\label{eq:downsampler}
\end{equation}
In essence, the function in~\eqref{eq:downsampler} implements a projection on a subspace of lower dimensionality. For these type we showed in~\cite{Geiger_LossLong} that the relative information loss is related to the information dimension of the output, which in our case is given by the number of its non-zero entries, i.e., by
\begin{equation}
 \infodim{\tilde{Y}_1^n} = \left\lfloor\frac{n}{M}\right\rfloor.
\end{equation}
With $\infodim{X_1^n}=n$ and by the fact that $\lfloor n/M\rfloor = n/M+\{n/M\}$, where $\{\cdot\}$ denotes the fractional part, we obtain
\begin{multline}
 \limn \relLoss{X_1^n\to \tilde{Y}_1^n} = 1-\limn\frac{n/M+\{n/M\}}{n}\\
= 1-\frac{1}{M} = \frac{M-1}{M}.
\end{multline}
The second equality follows because the magnitude of the fractional part is bounded by unity and that, thus, this term vanishes in the limit.

\section{Conclusion}
In this work we extended previous results about the information loss induced by deterministic, memoryless input-output systems from random variables to stationary stochastic processes with continuous distribution. Notably, we showed a connection between the rate of information loss and the differential entropy rates of the input and output processes for a special class of functions. While redundancy decreases the information loss rate for this class of systems, systems which destroy an infinite amount of information do not benefit from redundancy of the process in most practical cases. Future investigations shall focus on the extension to systems with memory and on the problem of reconstructing the input process.

As side results, we presented sufficient conditions for the Markovian input process and the system function such that the output process is Markov.

\ifthenelse{\arxiv=1}{
\appendices
\section{Proofs}
\subsection{Proof of Proposition~\ref{prop:equalityThenLump}}
Note that a possible definition of Markovity is given by
\begin{definition}[Markov Process{~\cite[II.6, p.~80]{Doob_StochasticProcesses}}]\label{def:Markov}
 A process $\Xvec$ is a Markov process iff for all $i\in\mathbb{N}$, $a\in\mathbb{R}$, and integers $n_1<n_2<\dots<n_i<n$, with probability one,
\begin{multline}
\Prob{X_n\le a|X_{n_1}=x_{n_1},\dots,X_{n_i}=x_{n_i}}\\
=\Prob{X_n\le a|X_{n_i}=x_{n_i}}.\label{eq:markovdef}
\end{multline}
\end{definition}

Clearly, a process is Markov if, for all $n$,
\begin{equation}
 f_{X_n|X_1^{n-1}}\stackrel{a.e.}{=}f_{X_n|X_{n-1}}\label{eq:equalPDFs}
\end{equation}
holds $P_{X_1^{n-1}}$-a.s. because~\eqref{eq:markovdef} results from integrating the densities over $(-\infty,a]$.

The proof of the proposition follows along the same lines as the proof for Markov chains given in~\cite{GeigerTemmel_kLump}, and is built on the following Lemma, which is an extension of~\cite[Thm.~4.5.1]{Cover_Information2}:

\begin{lem}[Bounds on the differential entropy rate]\label{lem:boundsDERate}
Let $\Xvec$ be a stationary Markov process with differential entropy rate $\derate{X}=\diffent{X_2|X_1}$ and let $\Yvec$ be a stationary process derived from $\Xvec$ by $Y_n=g(X_n)$. Then,
\begin{equation}
 \diffent{Y_n|Y_2^{n-1},X_1} \le \derate{Y} \le \diffent{Y_n|Y_1^{n-1}}.
\end{equation}
\end{lem}

\begin{IEEEproof}
 The upper bound follows from the fact that conditioning reduces entropy, so we only have to show the lower bound. For this, note that by Markovity of $\Xvec$,
\begin{equation}
 \diffent{Y_n|Y_2^{n-1},X_1} = \diffent{Y_n|Y_2^{n-1},X_k^1}
\end{equation}
for all $k<1$. Let $U_k=(Y_2^{n-1},X_k^1)$ and $V_k=Y_k^{n-1}$. Obviously, there exists a function $f$ such that $V_k=f(U_k)$, namely the function which is the identity function on the first $n-2$, and the function $g$ on the last $2-k$ elements. By showing that
\begin{equation}
 \diffent{Y_n|U_k} \le \diffent{Y_n|V_k}
\end{equation}
the lower bound is proved by~\cite[Thm.~14.7]{Papoulis_Probability}
\begin{multline}
 \diffent{Y_n|Y_2^{n-1},X_1}=\lim_{k\to-\infty}\diffent{Y_n|U_k}\\
 \le \lim_{k\to-\infty}\diffent{Y_n|V_k}=\derate{Y}.
\end{multline}

Thus, we write
\begin{align}
& \diffent{Y_n|V_k}-\diffent{Y_n|U_k}\notag\\
&=\diffent{Y_n,V_k}-\diffent{V_k}-\diffent{Y_n,U_k}+\diffent{U_k}\\
&\stackrel{(a)}{=} \ent{U_k|V_k}-\expec{\log|\det\Jac{f}{U_k}|}\notag\\
&\quad{}-\ent{U_k,Y_n|V_k,Y_n}+\expec{\log|\det\Jac{f}{U_k}|}\\
&=\ent{U_k|V_k}-\ent{U_k|V_k,Y_n}\\
&\ge 0
\end{align}
where $(a)$ is due the multivariate extension of~\eqref{eq:lossdiffent} and since the determinant of the Jacobian matrix is the same for the function $f$, and for a function which applies $f$ to some, and the identity function to the rest of the elements. This completes the proof.
\end{IEEEproof}

We now turn to the
\begin{IEEEproof}[Proof of Proposition~\ref{prop:equalityThenLump}]
Note that the assumption implies that
\begin{multline}
 \int_{\dom{X}}\int_{\dom{Y}} f_{Y_2,X_1}(y,x) \log\left(\frac{f_{Y_2|X_1}(y|x)}{f_{Y_2|Y_1}(y|g(x))}\right)dydx\\
=\diffent{Y_2|Y_1}-\diffent{Y_2|X_1}=0
\end{multline}
which renders the upper bounds of Lemma~\ref{lem:boundsDERate} equal for $n=2$. Thus, $\derate{Y}=\diffent{Y_n|Y_1^{n-1}}=\diffent{Y_2|Y_1}$ for all $n$. By stationarity, 
\begin{align}
0&= \diffent{Y_n|Y_{n-1}}-\diffent{Y_n|Y_1^{n-1}}\\
&=\mutinf{Y_n;Y_1^{n-2}|Y_{n-1}}\\
&= \expec{ \log \left(\frac{f_{Y_n,Y_1^{n-2}|Y_{n-1}}(Y_1^n)}{f_{Y_n|Y_{n-1}}(Y_{n-1}^n)f_{Y_1^{n-2}|Y_{n-1}}(Y_1^{n-1})}\right)}\\
&= \expec{ \log \left(\frac{f_{Y_n|Y_1^{n-1}}(Y_1^n)}{f_{Y_n|Y_{n-1}}(Y_{n-1}^n)}\right)}\\
&= \expec{\kld{f_{Y_n|Y_1^{n-1}}(\cdot,Y_1^{n-1})}{f_{Y_n|Y_{n-1}}(\cdot,Y_{n-1})}}
\end{align}
where $\kld{\cdot}{\cdot}$ is the Kullback-Leibler divergence and where in the last line the expectation is taken w.r.t. $Y_1^{n-1}$. 

The expectation of a non-negative RV, such as the Kullback-Leibler divergence above, can only be zero if this RV is almost surely zero. Together with the fact that the Kullback-Leibler divergence between two PDFs vanishes iff the PDFs are equal almost everywhere, the assumption of the proposition implies that
\begin{equation}
f_{Y_n|Y_1^{n-1}}\stackrel{a.e.}{=}f_{Y_n|Y_{n-1}}
\end{equation}
$P_{Y_1^{n-1}}$-a.s. But this implies Markovity by Definition~\ref{def:Markov} (cf.~\eqref{eq:equalPDFs}) and completes the proof.
\end{IEEEproof}

\subsection{Proof of Corollary~\ref{cor:sufficientConditions}}
Note that~\eqref{eq:condEq} implies $f_{Y_2|X_1}(y_2|x)=f_{Y_2|X_1}(y_2|x')$ for all $x,x'$ within the support of $f_X$. Now
\begin{equation}
 f_{Y_2|Y_1}(y_2|y_1)=\frac{1}
	{ f_Y(y_1) }\sum_{x_1\in\preim{y_1}} \frac{f_{Y_2|X_1}(y_2|x_1)f_X(x_1)}{|g'(x_1)|}.
\end{equation}
Let $g_+^{-1}[y_1]:=\{x\in\preim{y_1}:f_X(x)>0\}$ and let $\hat{x}$ be an arbitrary element of this set. We proceed
\begin{align}
&f_{Y_2|Y_1}(y_2|y_1)\notag\\
&=\frac{1}{f_Y(y_1)}\sum_{x_1\in g_+^{-1}[y_1]} \frac{f_{Y_2|X_1}(y_2|x_1)f_X(x_1)}{|g'(x_1)|}\\
&\stackrel{(a)}{=} \frac{f_{Y_2|X_1}(y_2|\hat{x})}{f_Y(y_1)}\sum_{x_1\in g_+^{-1}[y_1]} \frac{f_X(x_1)}{|g'(x_1)|}\\
&=f_{Y_2|X_1}(y_2|\hat{x})
\end{align}
where $(a)$ is due to~\eqref{eq:condEq}. Since $f_{Y_2,X_1}=f_{Y_2|X_1}f_X$ we can apply Proposition~\ref{prop:equalityThenLump} to complete the first part of the proof.

For the second part, note that we have with Proposition~\ref{prop:equalityThenLump}
\begin{equation}
 \derate{Y}=\diffent{Y_2|X_1}
\end{equation}
and thus, with Proposition~\ref{prop:lossratePBFs} and~\eqref{eq:lossdiffent},
\begin{multline}
 \lossrate{X\to Y} = \diffent{X_2|X_1}-\diffent{Y_2|X_1}+\expec{\log|g'(X)|}\\
 = \ent{X_2|X_1,Y_2}.
\end{multline}
It remains to show that~\eqref{eq:tightEq} implies equality in~\eqref{eq:prop52ndeq} in the proof of Proposition~\ref{prop:UBMarkov}. To this end, observe that
\begin{equation}
 \ent{W_2|X_1}-\ent{W_2|X_1,Y_2} = \mutinf{W_2;Y_2|X_1}
\end{equation}
vanishes if for all $y\in\dom{Y}$ and all $x\in\dom{X}$ such that $f_X(x)>0$ and for all $w$ such that $\Prob{W_2=w|X_1=x}>0$
\begin{equation}
f_{Y_2|X_1}(y|x)=f_{Y_2|W_2,X_1}(y|w,x).
\end{equation}
But
\begin{equation}
f_{Y_2|W_2,X_1}(y|w,x)
=\frac{f_{X_2|X_1}(g_w^{-1}(y)|x)}{p(w|x)|g'(g_w^{-1}(y))|}
\end{equation}
where $p(w|x)=\Prob{W_2=w|X_1=x}$. Let, for a given $x$, $\hat{w}$ satisfy $p(\hat{w}|x)>0$.
The proof is completed by recognizing that
\begin{align}
&f_{Y_2|X_1}(y|x)\notag\\
&=\sum_w p(w|x) f_{Y_2|W_2,X_1}(y|w,x)\\
&=\sum_w \frac{f_{X_2|X_1}(g_w^{-1}(y)|x)}{|g'(g_w^{-1}(y))|}\\
&\stackrel{(a)}{=} \frac{f_{X_2|X_1}(g_{\hat{w}}^{-1}(y)|x)}{|g'(g_{\hat{w}}^{-1}(y))|}\sum_w [p(w|x)>0]\\
&=\frac{f_{X_2|X_1}(g_{\hat{w}}^{-1}(y)|x)}{|g'(g_{\hat{w}}^{-1}(y))|}%\notag\\ &\quad\times
\ \card{\{w:p(w|x)>0\}}\\
&\stackrel{(b)}{=} \frac{f_{X_2|X_1}(g_{\hat{w}}^{-1}(y)|x)}{|g'(g_{\hat{w}}^{-1}(y))|} \frac{1}{p(\hat{w}|x)}\\
&=f_{Y_2|W_2,X_1}(y|\hat{w},x)
\end{align}
where $(a)$ is due to~\eqref{eq:tightEq:a} and $(b)$ is due to~\eqref{eq:tightEq:b}.
\endproof     

\subsection{Proof of Proposition~\ref{prop:RelLossgeRelLossRate}}
We start with showing $\relLoss{X\to Y} =P_X(\dom{X}_c)$. To this end, letting $\infodim{Z}$ denoting the R\'{e}nyi information dimension of $Z$~\cite{Renyi_InfoDim} and employing~\cite{Wu_Renyi,Smieja_EntropyMixture}, we write
\begin{equation}
 \infodim{X|Y=y}
 = \sum_{i=1}^K \infodim{X|Y=y,X\in\dom{X}_i} P_{X|Y=y}(\dom{X}_i)
\end{equation}
where $K=\card{\{\dom{X}_i\}}$. W.l.o.g., the partition is indexed such that the first $L$ elements correspond to subsets $\dom{X}_i$ on which $g$ is constant. Thus, $\dom{X}_c=\bigcup_{i=1}^L\dom{X}_i$. It follows for $i>L$, from the bijectivity of $g_i$, that $\infodim{X|Y=y,X\in\dom{X}_i}=0$. Moreover, if for $i\le L$ we have $X\in\dom{X}_i$, it follows that $Y=c_i$, and that thus
\begin{align}
 &\infodim{X|Y=y}\notag\\
 &= \sum_{i=1}^L \infodim{X|Y=y,X\in\dom{X}_i} P_{X|Y=y}(\dom{X}_i)\\
 &= \sum_{i=1}^L \infodim{X|X\in\dom{X}_i} P_{X|Y=y}(\dom{X}_i)\\
 &\stackrel{(a)}{=} \sum_{i=1}^L P_{X|Y=y}(\dom{X}_i)\\
 &\stackrel{(b)}{=} P_{X|Y=y}(\dom{X}_c).
\end{align}
where $(a)$ is due to the fact that the RV $X$ restricted to the non-empty set $\dom{X}_i$ possesses a density and $(b)$ follows from the fact that the partition consists only disjoint sets.
\begin{figure*}
\begin{equation}\label{eq:longFormula}
 \begin{split}
 &\relLoss{X_1^n\to Y_1^n}= P_{X_1^n}(\dom{X}_c^n)\\
& +\frac{n-1}{n}P_{X_1^n}(\dom{X}_c^{n-1}\times\overline{\dom{X}_c})+\frac{n-1}{n}P_{X_1^n}(\dom{X}_c^{n-2}\times\overline{\dom{X}_c}\times\dom{X}_c)+\cdots + \frac{n-1}{n}P_{X_1^n}(\overline{\dom{X}_c}\times\dom{X}_c^{n-1})\\
&\vdots\\
&+\frac{1}{n}P_{X_1^n}(\dom{X}_c\times\overline{\dom{X}_c}^{n-1})+\frac{1}{n}P_{X_1^n}(\overline{\dom{X}_c}\times\dom{X}_c\times\overline{\dom{X}_c}^{n-2})+\cdots+\frac{1}{n}P_{X_1^n}(\overline{\dom{X}_c}^{n-1}\times\dom{X}_c)\\
&+\frac{0}{n}P_{X_1^n}(\overline{\dom{X}_c}^{n})
\end{split}
\end{equation}
\hrule
\end{figure*}
We now combine $\infodim{X}=1$ and
\begin{equation}
 \infodim{X|Y}=\int_\dom{Y} \infodim{X|Y=y} dP_Y(y) = P_{X}(\dom{X}_c)
\end{equation}
with the fact that $\relLoss{X\to Y}=\frac{\infodim{X|Y}}{\infodim{X}}$~\cite{Geiger_LossLong} and obtain the first part of the proof.

Now take a finite sequence $X_1^n$ obtained from the stochastic process $\Xvec$ and look at the relative information loss incurred in $g$. Similarly as in the proof of Proposition~\ref{prop:WBound}, $g^n$ induces a finite partition of $\dom{X}^n$. Moreover, for every element of this partition, $g^n$ is a composition of a bijective, differentiable function and, possibly, a projection. We can thus apply the result about dimensionality reduction presented in~\cite{Geiger_LossLong} which leads to~\eqref{eq:longFormula} where $\overline{\dom{X}_c}=\dom{X}\setminus\dom{X}_c$. Compactly written, we get
\begin{equation}
 \relLoss{X_1^n\to Y_1^n}=\frac{1}{n}\sum_{i=1}^n i \Prob{\card{\{X_j\in X_1^n:X_j\in\dom{X}_c\}}=i}.
\end{equation}
Defining
\begin{equation}
 V_n:=\begin{cases}
      1,& \text{if } X_n\in\dom{X}_c\\0,&\text{else}
     \end{cases}
\end{equation}
and $Z_n:=\sum_{j=1}^n V_j$, and with the linearity of expectation we get
\begin{align}
 &\relLoss{X_1^n\to Y_1^n}\notag\\
&=\frac{1}{n}\sum_{i=1}^n i \Prob{\sum_{j=1}^n V_j=i}\\
&=\frac{1}{n}\sum_{i=1}^n i \Prob{Z_n=i}\\
&=\frac{1}{n}\expec{Z_n}\\
&=\frac{1}{n}\sum_{j=1}^n \expec{V_j}\\
&\stackrel{(a)}{=}\expec{V}\\
&= P_X(\dom{X}_c)
\end{align}
where $(a)$ is due to stationarity of $\Xvec$. This completes the proof.
\endproof                                                        
}{}

\bibliographystyle{IEEEtran}
\bibliography{IEEEabrv,/afs/spsc.tugraz.at/project/IT4SP/1_d/Papers/InformationProcessing.bib,%
/afs/spsc.tugraz.at/project/IT4SP/1_d/Papers/ProbabilityPapers.bib,%
/afs/spsc.tugraz.at/user/bgeiger/includes/textbooks.bib,%
/afs/spsc.tugraz.at/user/bgeiger/includes/myOwn.bib,%
/afs/spsc.tugraz.at/user/bgeiger/includes/UWB.bib,%
/afs/spsc.tugraz.at/project/IT4SP/1_d/Papers/InformationWaves.bib,%
/afs/spsc.tugraz.at/project/IT4SP/1_d/Papers/ITBasics.bib,%
/afs/spsc.tugraz.at/project/IT4SP/1_d/Papers/HMMRate.bib,%
/afs/spsc.tugraz.at/project/IT4SP/1_d/Papers/ITAlgos.bib}

% \section*{Acknowledgments}
% The authors thank Yihong Wu, Wharton School, University of Pennsylvania, Siu-Wai Ho, Institute for Telecommunications Research, University of South Australia, and Sebastian Tschiatschek, Signal Processing and Speech Communication Laboratory, Graz University of Technology, for fruitful discussions and suggesting material.

\end{document}